\documentclass[runningheads,a4paper]{llncs}

\usepackage{amssymb}
\usepackage{graphicx}
\newcommand{\kotak}{\rule{.08in}{.08in}}
\newcommand{\kotakt}{\rule{.08in}{.12in}}

\usepackage{url}
\urldef{\mailsa}\path| mahyunst@yahoo.com |
\newcommand{\keywords}[1]{\par\addvspace\baselineskip
\noindent\keywordname\enspace\ignorespaces#1}

\begin{document}

\mainmatter  

\title{Simple Search Engine Model: Selective Properties}

\titlerunning{Simple Search Engine Model: Selective Properties}

\author{Mahyuddin K. M. Nasution$^{1}$%
\thanks{A draft}}%

\authorrunning{M. K. M. Nasution}

\institute{$^{1}$Departemen Teknologi Informasi, FASILKOM-TI,\\
Universitas Sumatera Utara, Padang Bulan 20155 USU, Medan, Indonesia.\\
\mailsa\\}

\toctitle{draf}
\tocauthor{}
\maketitle

\begin{abstract}
In this paper we study the relationship between query and
search engine by exploring the selective properties based on a simple
search engine. We used the set theory and utilized the words and terms for
defining singleton and doubleton in the event spaces and then provided their implementation for proving the existence of the shadow of  micro-cluster.
\keywords{singleton space, doubleton space, search term, query, triplet}
\end{abstract}

\section{Introduction}

Numerious studies of natural language processing (NLP) and Semantic Web utilize a search engine, mainly to obtain a set of documents, mainly to obtain a set of documents that include a given query and to get statistical information about an object such as entity name in hit count \cite{nasution2011a} by the search engine, but to bring the NLP and Semantic Web to life such as the information processing services provide the knowledge, for example: ontology construction,  knowledge extraction, question answering, and other purposes, all need more effort. However, to produce the enhanced relationship between a search engine and a query as novel property, we already defined some instances about simple search engines: singleton \cite{nasution2012a} and doubleton \cite{nasution2012b} spaces. This model based on the simple architecture of search engine for representing the collection of documents in general such that this model can distinguish the features of Web documents, whereby any query can gives the essential purpose of Information Retrieval \cite{nasution2012} strategies to meet user needs. Therefore, this paper aims to address some properties based on relations between singleton and doubleton in a triplet. We also provided the basis of this model to the micro-cluster for implementing the adaptive properties. 

\section{Some Terminologies}

We defined some terminologies as follows \cite{nasution2011a,nasution2012a,nasution2012b}.

\begin{definition}
\label{def:term}
A term $t_x$ consists of at least one or a set of words in a pattern, or $t_k = (w_1w_2\dots w_l)$, $l\leq k$, $k$ is a number of parameters representing word $w$, $l$ is the number of tokens (vocabularies) in $t_k$, $|t_k| = k$ is size of $t_k$. \kotak
\end{definition}

\begin{definition}
\label{def:searchengine}
Let a set of web pages indexed by search engine be $\Omega$, i.e., a set contains ordered pair of the terms $t_{k_i}$ and the web pages $\omega_{k_j}$, or $(t_{k_i},w_{k_j})$, $i=1,\dots,I$, $j = 1,\dots,J$. The relation table that consists of two columns $t_k$ and $\omega_k$ is a representation of $(t_{k_i},\omega_{k_j})$ where $\Omega_k = \{(t_k,\omega_k)_{ij}\}\subset\Omega$ or $\Omega_k = \{\omega_{k_1},\dots,\omega_{k_j}\}$. The cardinality of $\Omega$ is denoted by $|\Omega|$. \kotak
\end{definition}

\begin{definition}
\label{def:singleton}
Let $t_x$ is a search term, and $t_x\in{\cal S}$ where ${\cal S}$ is a set of singleton search term of search engine. A vector space $\Omega_x\subseteq\Omega$ is a singleton search engine event (\emph{singleton space of event}) of web pages that contain an occurrence of $t_x\in\omega_x$. The cardinality of $\Omega_x$ is denoted by $|\Omega_x|$. \kotak
\end{definition}

\begin{definition}
\label{def:doubleton}
Let $t_x$ and $t_y$ are two different search term, $t_x\ne t_y$, $t_x,t_y\in{\cal S}$, where ${\cal S}$ is a set of singleton search term of search engine. A doubleton search term is ${\cal D} = \{\{t_x,t_y\}: t_x,t_y\in\Sigma\}$ and its vector space denoted by $\Omega_x\cap\Omega_y$ is a double search engine event (\emph{doubleton space of event}) of web pages that contain a co-occurrence of $t_x$ and $t_y$ such that $t_x,t_y\in\omega_x$ and $t_x,t_y\in\omega_y$, where $\Omega_x, \Omega_y, \Omega_x\cap\Omega_y\subseteq\Omega$. \kotak
\end{definition}

Some adaptive properties are defined to know the efficient ways to access information by using simple search engine model. In general, all adaptive properties is to adopt the meaning of singleton and doubleton in equations as follows
\[
|\Omega_x| = |\Omega_x|+\Omega_y|
\]
and
\[
|\Omega_x\cap\Omega_y| = |\Omega_x\cap\Omega_y|+|\Omega_x\cap\Omega_x|+|\Omega_y\cap\Omega_y|
\]
Therefore, statistically either singleton or the doubleton contain bias information. 

Usually to improve the quality of statistical information by a search engine of a given query, the count is processed statistically based on above properties. However, to make an additional improvement, we must devote more attention to results of search engine and carefully handle the count for developing the selective model. 

\section{The Selective Properties}

The purpose of selective properties is to construct an approach for eliminating bias by using the selected results of simple search engine. One of results by a search engine as follows \cite{nasution2011c}.

\begin{definition}
\label{def:snippet}
Let $t_x$ is a search term. $S = \{w_1,\dots,w_{max}\}$ is a Web snippet (briefly snippet), $S\subset\omega_{x_i}\in\Omega$, where $max\leq 50$ words to the left and right of $t_x$ that returned by any search engine. $L = \{S_i:i=1,\dots,n\}$ is a list of snippets. \kotak
\end{definition}

\subsection{Triplet}

A construction of relationship based on frequency of words between search term, snippets, and words is as follows \cite{nasution2011b}.

\begin{definition}
\label{def:triplet}
A relationship between search term, web snippets and words is defined as the mixture $p(t_o,S,w) = t_o\times S\times w$, $t_o\in O$, $S\in L\subseteq\Omega$, $w\in S$. A vector space of $P(t_o,S,w)$ is defined as ${\bf w} = \{w_i,\dots,w_j\}$. \kotak
\end{definition}

$ = [\nu_i,\dots,\nu_j]$, $\nu_i\geq\dots\geq\nu_j$, where $w_i,\dots,w_j$ are the unique words in $S$ and $\nu_i,\dots,\nu_j$ are the weights of word.

The relations of the search term and the Web snippets and the words, we called it as triplet, or we rewrote as a term-snippet-word. The triplet is a base for exploring features of: Web pages or Web documents. The features exploration is to describe an object literally in text if the purpose of search term is to explain the object. A relation between term and snippet logically is $|t_0\cap S| = 1$ if $t_o\in S$ and $= 0$ otherwise, and 
\begin{equation}
P(t_o\cap S) = \frac{1}{2},
\end{equation}
or probability of the search term in list of snippets are
\begin{equation}
P(t_o\cap L) = \frac{\sum_{i=1}^n\frac{1}{2}}{n}.
\end{equation}
A relation of snippet-word interpretated as follows,
\begin{equation}
\label{eq:sw}
P(S\cap {\bf w}) = \sum_{j=1}^m \frac{1}{max}
\end{equation}
where $m$ is a number of same word in vocabulary, or probability of the word in list of snippets is as follows
\begin{equation}
\label{eq:lw}
P(L\cap {\bf w}) = \sum_{i=1}^n \sum_{j=1}^m\frac{1}{max_i},
\end{equation}
while the term-word has two representations logically, i.e. $|t_o\cap w| = \sum_{j=1}^m \frac{1}{max}$ if $t_0\in S$ and $=0$ if $t_0\not\in S$, or probability of $t_o\cap w$ in $S$ as follows
\begin{equation}
P(t_o\cap {\bf w}) = \frac{\sum_{j=1}^m\frac{1}{max}}{2}.  
\end{equation} 
Probability of $t_o\cap w$ in $L$ to be
\[
P(t_o\cap {\bf w})_L = \sum_{i=1}^n\frac{\sum_{j=1}^m\frac{1}{max}}{2}
\]
\begin{equation}
\label{eq:tw}
P(t_o\cap {\bf w})_L = \sum_{i=1}^n\sum_{j=1}^m \frac{1}{2~max_i}
\end{equation}

Trivially for each snippet there exists a set of words ${\bf w} = \{w_i|i=1,\dots,n\}$, i.e. ${\bf w}$ contains at least one word of search term or the search term self. 

\begin{definition}
Word frequency is a word number uniquely in a set of words, i.e. $\exists\nu\in\Re$ $\forall w \in {\bf w}$, $\Re$ is a real number set, and $\nu\in\Re$ as a weight of word. Generally, there is $1:1$ function $\varpi$ such that 
\begin{equation}
\label{eq:wf}
\varpi : {\bf w} \rightarrow \Re
\end{equation}
In this case, $\Re$ as a vector space of ${\bf w}$. \kotak
\end{definition}

\begin{lemma}
\label{lem:es1}
If a set of words is representation of snippets in list of snippets, then vector space of words set contains probability of word in snippets.
\end{lemma}
\begin{proof}
In Eq. (\ref{eq:sw}) as probability of word based on frequency of word in snippet where $m$ is number of word uniquely in snippet. Therefore, Eq. (\ref{eq:lw}) also is probability of word based on frequency of word in list of snippets. Reasonally, because $|t_o\cap {\bf w}| \ne 0$. Based on Eq. (\ref{eq:tw}) and Eq. (\ref{eq:wf}) we have
\[
\nu = \varpi(w) = \sum_{i=1}^n\sum_{j=1}^m \frac{1}{2~max_i}
\]
as probability of word in a list of snippets for a search term, and $\Re$ contains the value of Eq. (\ref{eq:tw}) for all $w\in{\bf w}$.  \kotakt
\end{proof}

\begin{lemma}
\label{lem:es2}
If a set of words is representation of snippets in list of snippets, then the set of words is an event space.
\end{lemma}
\begin{proof}
As a search term each word in ${\bf w}$ based on Definition \ref{def:singleton} has $\Omega_w$, and each two words in pair based on Definition \ref{def:doubleton} has $\Omega_{w_i}\cap\Omega_{w_j}$. Thus ${\bf w}$ be an event space. \kotakt
\end{proof}

The event space contains vectors $\mu_i = |\Omega_{w_i}|$, $i=1,\dots,n$, and we called it as singleton event space.

\begin{proposition}
If $p(t_o,S,w)$ is a triplet for $t_o$, then there are at least one vector space of $p(t_o,S,w)$.
\end{proposition}
\begin{proof}
The direct consequence of Lemma \ref{lem:es1} and Lemma \ref{lem:es2}. \kotakt
\end{proof}

Thus, based on Lemma \ref{lem:es1} and Lemma \ref{lem:es2}, there are two vector spaces for a list of snippets, i.e.
\begin{definition}
A set of words ${\bf w}$ is a context if ${\bf w}$ has two vector spaces such that satisfies
\begin{enumerate}
\item for $[\nu_i,\dots,\nu_j]$ as vector space,  $\nu_i\geq\dots\geq\nu_j$, and
\item for $[\mu_i,\dots,\mu_j]$ as singleton vector space, 
$\mu_i\geq\dots\geq\mu_j$. \kotak
\end{enumerate} 
\end{definition}

\subsection{Micro-cluster}

In part of section, we define the words undirected graph $G = (V,E)$ to describe the relations between words \cite{nasution2010}. 

\begin{definition}
\label{def:mikrokluster}
Assume a sub-graph $G'$, $G'\subset G$, $G'$ is a \emph{micro-cluster} satisfies the conditions as follows
\begin{enumerate}
\item There are a set of word ${\bf w} = \{w_i,\dots,w_j\}$ whose vector space $[\nu_i,\dots,\nu_j]$ and $\nu_i\geq\dots\geq\nu_j\geq\alpha$, where $\alpha$ is a threshold.
\item There are an one-one function $f : {\bf w}\rightarrow V$ such that $f(w) = v$, $\forall w\in{\bf w}\exists v\in V$ where $v\in V$ is a vertex in $G'$.
\item There are an one-one function $\rho : {\bf w}\times{\bf w}\rightarrow E$ such that $\rho(w_i,w_j)=e$, $\forall w_i,w_j\in{\bf w}$, where $\rho$ is a relation among words and $e\in E$ is a edge in $G'$.
\end{enumerate}
The micro-cluster is denoted by $G' = \langle V,E,{\bf w},f,\rho,\alpha\rangle$. \kotak
\end{definition}

A micro-cluster is maximal clique sub-graph of entity name where the node represents word that the highest score in document. However, let there is a set of words ${\bf w}$ whose weights above the threshold, the collection of words do not exactly refer to the same entity. To group the words into the appropriated cluster, we construct the trees of words. This based on an assumption that the words are that appear in same domains having closest relation. The tree is an optimal representation of relation in graph $G$.

\begin{definition}
\label{def:optimalmikrokluster}
A tree $T$ is an \emph{optimal micro-cluster} if and only if $T$ is a sub-graph of micro cluster $G'$, and is denoted by $T = \langle V_T,E_T,{\bf w}_T,f,\rho,\alpha\rangle$, where $V_T\subseteq V$, $E_T\subset E$, and ${\bf w}_T\subseteq {\bf w}$. \kotak
\end{definition}

In building the optimal micro-cluster, we save the strongest relations in $T$ between a word and another in $G'$ until $T$ has no cycle. A cycle is a sequence of two or more edges $(v_i,v_j),(v_j,v_k),\dots,(v_{k+1},v_i) \in E$ such that there is an optimal edge $(v_i,v_j)\in E$ connects both ends of sequence. Let a word is introduced as intrusive word about an entity, and there are at least one word of optimal micro-cluster has strongest relations with the entity, and an optimal micro-cluster is a group of words refer to that entity. However, the overlap keyword also exists in the same list. We define a strategy to select a relevant keywords among all list candidates. In this case, there are a few potential keywords for identifying the entity name.  

\begin{definition}
\label{def:mirorshade}
A vector space ${\bf s} = [|{\bf w_i}|,\dots,|{\bf w_j}|]$ is a \emph{mirror shade} of micro-cluster $G'$ if there is an one-one function $g : {\bf w}\rightarrow{\bf s}$, where ${\bf w_i},\dots,{\bf w_j}$ are in event space. Let ${\bf z}$ is a vector whose greatest value in ${\bf s}$, the vector space in range of $[0,1]$ is relatively defined as ${\bf s}_{[0,1]} = [|{\bf w_i}|/|{\bf z}|,\dots,|{\bf w_j}|/|{\bf z}|] = [\mu_i,\dots,\mu_j]$.\kotak
\end{definition}

We also can generate for example another vectors from $\Omega_i,\dots,\Omega_j$ for words $w_i,\dots,w_j$ respectively such that $[\mu_i,\dots,\mu_j] = [\Omega_i,\dots,\Omega_j]$ is a mirror shade of $[\nu_i,\dots,\nu_j]$ from a set of word frequencies.

\begin{theorem}
\label{theo:optimalmikrokluster}
Let ${\bf s}_T\subseteq{\bf s}$, then ${\bf s}_T$ is the mirror shade of an optimal micro-cluster $T$.
\end{theorem}
\begin{proof}
Let ${\bf s}_T\subseteq {\bf s}$, based on Definition \ref{def:mirorshade} we have ${\bf w}_T\subseteq {\bf w}$, i.e. $g({\bf w}_T) = {\bf s}_T$ or because of $g$ is one-one function, $g^{-1}({\bf s}_T) = {\bf w}_T\subseteq{\bf w}$. Next, by applying Definition \ref{def:mikrokluster}, $f({\bf w}_T) = V_T$, or because of $f$ is one-one function, $f^{-1}(V_T) = {\bf w}_T\subseteq {\bf w}$, and $s_T = g({\bf w}_T) = g(f^{-1}(V_T))=f^{-1}g(V_T)$ and we obtain $\rho({\bf w},{\bf w}) = \rho(f^{-1}(V),f^{-1}(V))\subseteq E$, so $\rho({\bf s}_T\times{\bf s}_T) = \rho(g({\bf w}_T)\times g({\bf w}_T)) =\rho(g(f^{-1}(V_T))\times g(f^{-1}(V_T))) = \rho(f^{-1}g(V_T))\times f^{-1}g(V_T))) = f^{-1}g\rho(V_T\times V_T) = f^{-1}g(\rho(V_T\times v_T))$ because of $f^{-1}g$ is also one-one function, this means that $V_T\subseteq V$ has ${\bf s}_T$ as a mirror shade of ${\bf w}_T$. \kotakt
\end{proof}

\section{Conclusions and Future Work}
The selective properties have been derived from the singleton and doubleton based on the tiplet concept. Through these properties have been proven the existence of the shadow of any micro-clusters for the space of events. Our future work is about the relation between adaptive and selective properties for exploring an overlap principle.

\end{document}